\documentclass[aps,a4paper,superscriptaddress,twocolumn,showpacs,preprintnumbers,amsmath,amssymb,pra]{revtex4-1}

\usepackage{graphicx}
\usepackage{dcolumn}
\usepackage{bm}
\usepackage{amsmath}
\usepackage{amsthm}
\usepackage{hyperref}
\usepackage{dsfont}
\usepackage{color}

\newtheorem{theorem}{Theorem}

\newtheorem{definition}[theorem]{Definition}
\newtheorem{proposition}[theorem]{Proposition}
\newtheorem{lemma}[theorem]{Lemma}
\newtheorem{example}[theorem]{Example}

\begin{document}


\title{Near-resonant optical forces beyond the two-level approximation for a continuous source of spin-polarized cold atoms}
\author{T.~Vanderbruggen}
\email{thomas.vanderbruggen@icfo.es}
\affiliation{ICFO - Institut de Ciencies Fotoniques, Mediterranean Technology Park, 08860 Castelldefels (Barcelona), Spain}
\author{M.~W.~Mitchell}
\affiliation{ICFO - Institut de Ciencies Fotoniques, Mediterranean Technology Park, 08860 Castelldefels (Barcelona), Spain}
\affiliation{ICREA - Instituci{\'o} Catalana de Recerca i Estudis Avan\c cats, 08015 Barcelona, Spain}

\date{\today}

\begin{abstract}
We propose a method to generate a source of spin-polarized cold atoms which are continuously extracted and guided from a magneto-optical trap using an atom-diode effect. We show that it is possible to create a pipe-like potential by overlapping two optical beams coupled with the two transitions of a three-level system in a ladder configuration. With alkali-metal atoms, and in particular with $^{87}$Rb, a proper choice of transitions enables both the potential generation and optical pumping, thus polarizing the sample in a given Zeeman state. We extend the Dalibard and Cohen-Tannoudji dressed-atom model of radiative forces to the case of a three-level system. We derive expressions for the average force and the different sources of momentum diffusion in the resonant, non-perturbative regime. We show using numerical simulations that a significant fraction of the atoms initially loaded can be guided over several centimeters with output velocities of a few meters per second. This would produce a collimated continuous source of slow spin-polarized atoms suitable for atom interferometry.
\end{abstract}


\maketitle

\section{Introduction}

The achievement of laser cooling \cite{chu1998,*cohen1998,*phillips1998} resulted in an unprecedented control of atomic ensembles, which led to a dramatic evolution of atom-based sensors such as atomic clocks, inertial sensors, and magnetometers \cite{kitching2011}. These sensors are today state-of-the-art devices and their improvement is a major technological issue. 

However, the cold atom cloud as the source of the sensors needs to be reloaded after each measurement cycle. This limits the maximum repetition rate of the sensor and creates a sampling of the measured variable. It is, for example, responsible for the Dick effect that is nowadays limiting optical lattice clocks \cite{lodewyck2009}. The development of continuous cold atomic sources can thus lead to important improvements and new designs of atom interferometers. Continuous atomic sources are often associated with an atomic guide that is usually realized with a magnetic potential \cite{roos2003,falkenau2011} and loaded from a cold atomic source. Nevertheless, the presence of a magnetic field may be a limitation for interferometric applications, and particularly for magnetometers.

Here we propose an all-optical method to create a continuous source of guided cold atoms loaded directly from a magneto-optical trap (MOT), as shown in Fig.~\ref{fig:schema}~(a). It relies on an atom-diode effect \cite{thorn2008,bennerman2009} that permits the atoms to enter into the guide but forbids them to escape. The method uses optical pumping to both keep the atoms in the guided state after loading and spin-polarize the sample in a given Zeeman sublevel. The guide thus provides a collimated and continuous source of polarized cold atoms suitable for atomic sensors.

\begin{figure}[!h]
\begin{center}
\includegraphics[width=8.5cm,keepaspectratio]{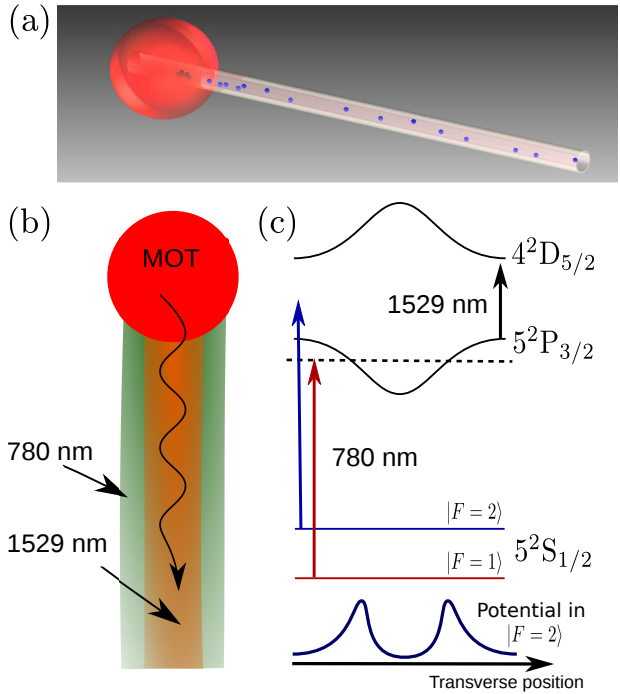}
\caption{(Color online) (a) Concept: A pipe-like potential is generated all-optically. It continuously extracts atoms from a magneto-optical trap and guides them. (b) The 780~nm and 1529~nm beams are coaxial, with the 780 nm beam being slightly larger. (c) Relevant atomic levels of $^{87}$Rb versus the transverse direction. The 1529~nm beam induces a position-dependent differential light-shift on the transition at 780~nm. The detuning of the 780~nm beam is chosen to generate a state-dependent dipole force which implements the atom diode. For an atom in $\left| F=2 \right\rangle$, the potential in the transverse direction is a double barrier, that is, a pipe, due to the cylindrical symmetry of the problem.}
\label{fig:schema}
\end{center}
\end{figure} 

The article analyses the main physical processes involved in the proposal implementation and is organized as follows. In a brief presentation of the method we explain how light-shift engineering in a three-level system leads to the continuous extraction of atoms from a MOT into an optical guide. We then model the forces for a three-level atomic system in a ladder configuration ($\Xi$-system): from the master equation in the Lindblad form, we derive the optical Bloch equations and introduce the dressed state picture which simplifies the description of the internal state dynamics to rate equations. Since the guide performance is limited by the heating rates, we quantify the various diffusion processes that an atom experiences when it propagates along the guide. We exhibit a possible implementation with $^{87}$Rb atoms, but the method can be generalized to any alkali-metal atom. We finally use the theoretical results previously obtained to simulate the stochastic trajectories of atoms in the guide solving Langevin-like equations. From these simulations we estimate the performance of the guide and show that the proposed method is viable.


\section{Principle of the method}


We now give a brief overview of how the guide is generated and continuously loaded, exhibit the various physical effects involved, and explain how they combine. 

The main idea is to coaxially overlap two beams [Fig.~\ref{fig:schema}~(b)] to engineer the internal state light-shifts and design the guiding potential. As shown in Fig.~\ref{fig:schema}~(c) for $^{87}$Rb atoms, the first beam at 780~nm is tuned to the blue of the $\left| 5^{2}\mathrm{S}_{1/2}, F=2 \right\rangle \rightarrow \left| 5^{2}\mathrm{P}_{3/2} \right\rangle$ transition, and to the red of the $\left| 5^{2}\mathrm{S}_{1/2}, F=1 \right\rangle \rightarrow \left| 5^{2}\mathrm{P}_{3/2} \right\rangle$ transition. The 780~nm field then realizes a state-dependent dipole potential: repulsive for atoms in $\left| F=2 \right\rangle$ and attractive for atoms in $\left| F=1 \right\rangle$. The second beam, placed on the red of the $5^{2}\mathrm{P}_{3/2} \rightarrow 4^{2}\mathrm{D}_{5/2}$ transition at 1529~nm, induces a differential light-shift on the transition at 780~nm \cite{brantut2008,bertoldi2010}. Therefore, the 780~nm radiation is far-detuned at the center of the beams, resulting in a strong reduction of the dipole potential. As a consequence, the potential is created only at the periphery of the beam, producing a state-selective barrier with a pipe shape due to the cylindrical symmetry of the problem. Moreover, the differential light-shift also strongly lowers the spontaneous emission rate inside the guide, reducing the heating rate and thus increasing the guided distance of the atoms. 


Due to the differential light-shift, the 780~nm light is resonant (or nearly resonant) with atoms in $\left| F=1 \right\rangle$ at some position inside the guide. Therefore, atoms in $\left| F=1 \right\rangle$ from the MOT can pass through the guide barrier before being repumped into the barrier-sensitive state $\left| F=2 \right\rangle$ by the 780~nm light. In other words, an atom can enter the guide but cannot leave it: this is the atom-diode effect which continuously extracts atoms from the MOT and fills the guide. We will see in Sec.~\ref{sec:doubly_closed_trans} that a proper choice of polarizations for the optical beams drives closed transitions that maintain the atoms in the guided state $\left| F=2 \right\rangle$ while optically pumping and polarizing the atomic sample.

For the atom-diode effect to occur the 780~nm radiation is not far-detuned from the $5^{2}\mathrm{S}_{1/2} \rightarrow 5^{2}\mathrm{P}_{3/2}$ transition and spontaneous emission will play an important role in the dynamics of the system. The scattering induces a heating of the sample that limits the guided distance for a given transport velocity. Moreover, the occupation of the excited states is not negligible near the barrier so that the dipole force of the 1529~nm light contributes to the overall potential. These effects may modify the na{\"i}ve description introduced above. Therefore, to verify whether the method is viable and what performance to expect, we developed a comprehensive model of the problem, expanding the Dalibard and Cohen-Tannoudji dressed-atom approach \citep{dalibard1985}. The following sections present a detailed study of the optical forces and diffusion coefficients for a three-level atomic system driven near resonance.

\section{Dipole forces in a three-level atomic system}

We estimate the dipole force and the related momentum dispersion coefficients in a $\Xi$-system (Fig.~\ref{fig:three_level_syst}). To solve this problem, we derive the optical Bloch equations for such a configuration and introduce the dressed-atom picture, providing a simplified description of the system when the secular approximation is valid.

\begin{figure}[!h]
\begin{center}
\includegraphics[width=3.5cm,keepaspectratio]{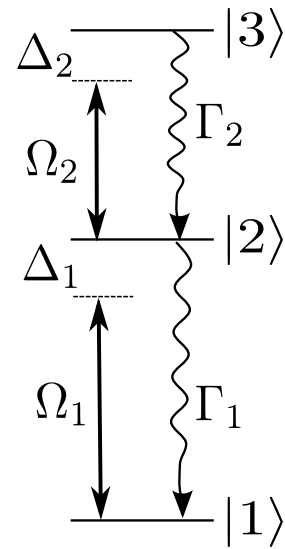}
\caption{Three-level system in a $\Xi$ configuration. The system is doubly-driven with electromagnetic fields: the $\left| 1 \right\rangle \leftrightarrow \left| 2 \right\rangle$ transition at $\omega_{0}^{(1)}$ is driven with a detuning $\Delta_{1} = \omega_{0}^{(1)} - \omega_{L}^{(1)}$ (note that with this sign convention $\Delta < 0$ corresponds to blue detuned light) and a Rabi frequency $\Omega_{1}$, whereas the $\left| 2 \right\rangle \leftrightarrow \left| 3 \right\rangle$ transition at $\omega_{0}^{(2)}$ is driven with a detuning $\Delta_{2} = \omega_{0}^{(2)} - \omega_{L}^{(2)}$ and a Rabi frequency $\Omega_{2}$. An atom in $\left| 2 \right\rangle$ can decay to $\left| 1 \right\rangle$ at a rate $\Gamma_{1}$ and an atom in $\left| 3 \right\rangle$ can decay to $\left| 2 \right\rangle$ at a rate $\Gamma_{2}$. We assume there is no decay from $\left| 3 \right\rangle$ to $\left| 1 \right\rangle$.}
\label{fig:three_level_syst}
\end{center}
\end{figure}

\subsection{Optical Bloch equations}

According to the notations in Fig.~\ref{fig:three_level_syst}, the Hamiltonian of the atom in the rotating frame is, in the $\left\{ \left| 1 \right\rangle,\left| 2 \right\rangle,\left| 3 \right\rangle \right\}$ basis,
\begin{equation}
H_{\rm at} = \hbar \left( 
\begin{array}{ccc}
0 & 0 & 0 \\
0 & \Delta_{1} & 0 \\
0 & 0 & \Delta_{2}
\end{array}
\right).
\end{equation}
By introducing the jump operators $\sigma_{12} = \left| 1 \right\rangle \left\langle 2 \right|$ and $\sigma_{23} = \left| 2 \right\rangle \left\langle 3 \right|$, the atom-light interaction can be modeled with the following Hamiltonian:
\begin{equation}
H_{\rm int} = \frac{\hbar \Omega_{1}}{2} \left( \sigma_{12} + \sigma_{12}^{\dagger} \right) + \frac{\hbar \Omega_{2}}{2} \left( \sigma_{23} + \sigma_{23}^{\dagger} \right),
\label{eq:ham_interaction}
\end{equation}
where the optical phases are assumed to be constant in time so that the Rabi frequencies $\Omega_{1}$ and $\Omega_{2}$ are real numbers. The Hamiltonian that describes the internal energy of an atom coupled to the two optical fields is thus $V = H_{\rm at} + H_{\rm int}$.

The internal state evolution of the system is governed by the following master equation:
\begin{equation}
\partial_{t} \rho = - \frac{i}{\hbar} \left[ V,\rho \right] + \Gamma_{1} \mathcal{L} \left[ \sigma_{12} \right] \rho + \Gamma_{2} \mathcal{L} \left[ \sigma_{23} \right] \rho,
\label{eq:master_eq}
\end{equation}
where the Lindblad superoperators that model the spontaneous emission processes are
\begin{equation}
\mathcal{L} \left[ \sigma \right] \rho = \sigma \rho \sigma^{\dagger} - \frac{1}{2} \left\{ \sigma \sigma^{\dagger} , \rho\right\},
\end{equation}
and $\left\{ A,B \right\} = AB+BA$ is the anticommutator. Using this master equation, the optical Bloch equations for a $\Xi$-system are derived in App.~\ref{app:opt_bloch_eqs}.

In the end, we are interested in the external motion of the atom which occurs at timescales much larger than the internal state evolution. We thus assume that the internal atomic state is always in the steady state $\rho_{\infty} \equiv \lim_{t \rightarrow \infty} \rho$. The steady state satisfies $\partial_{t} \rho_{\infty} = 0$ and its value is derived in App.~\ref{app:steady_rho}.

\subsection{Dressed-atom picture}

The dressed states are defined as the eigenstates of the Hamiltonian $V$, which admits the following eigendecomposition:
\begin{equation}
V = U^{\dagger} \widetilde{V} U,
\end{equation}
where $\widetilde{V} = \mathrm{diag} ( \widetilde{E}_{\alpha}, \widetilde{E}_{\beta}, \widetilde{E}_{\gamma} )$ is diagonal in the basis $\left\{ \left| \alpha \right\rangle, \left| \beta \right\rangle, \left| \gamma \right\rangle \right\}$ called the \textit{dressed states basis}. $U$ is the unitary transformation between the bare and dressed states basis $\left\{ \left| 1 \right\rangle, \left| 2 \right\rangle, \left| 3 \right\rangle \right\}\leftrightarrow\left\{ \left| \alpha \right\rangle, \left| \beta \right\rangle, \left| \gamma \right\rangle \right\}$ \cite{whitley1976}. 


A knowledge of the unitary transformation $U$ allows us to define the jump operators for the dressed states as $\widetilde{\sigma}^{\dagger} = U \sigma^{\dagger} U^{\dagger}$. Similarly, the density matrix in the dressed basis can be obtained from the solution $\rho$ of the Bloch equations in the bare basis: $\widetilde{\rho}= U \rho U^{\dagger}$.

The main advantage of the dressed-atom description is that when the system is strongly driven, the secular approximation applies and the equation of motion for the internal degrees of freedom reduce to rate equations, as shown in App.~\ref{app:trans_prob}. The decay rate between the dressed states $\left| k \right\rangle$ and $\left| k' \right\rangle$ is
\begin{equation}
\Gamma_{k \rightarrow k'} = \Gamma_{1} \left| \left( \widetilde{\sigma}_{12} \right)_{kk'} \right|^{2} + \Gamma_{2} \left| \left( \widetilde{\sigma}_{23} \right)_{kk'} \right|^{2}.
\end{equation}
The rate equations are
\begin{equation}
\partial_{t} \widetilde{\boldsymbol\rho} + \boldsymbol\Gamma \widetilde{\boldsymbol\rho} = \mathbf{0},
\label{eq:rate_eqns}
\end{equation}
where $\widetilde{\boldsymbol\rho} = \left( \widetilde{\rho}_{\alpha \alpha}, \widetilde{\rho}_{\beta \beta}, \widetilde{\rho}_{\gamma \gamma} \right)^{T}$ is the vector of the dressed state populations and the rate matrix is
\begin{equation}
\boldsymbol\Gamma = \left( 
\begin{smallmatrix}
\Gamma_{\alpha \rightarrow \beta} + \Gamma_{\alpha \rightarrow \gamma} & -\Gamma_{\alpha \rightarrow \beta} & -\Gamma_{\alpha \rightarrow \gamma} \\
-\Gamma_{\beta \rightarrow \alpha} & \Gamma_{\beta \rightarrow \alpha}+\Gamma_{\beta \rightarrow \gamma} & -\Gamma_{\beta \rightarrow \gamma} \\
-\Gamma_{\gamma \rightarrow \alpha} & -\Gamma_{\gamma \rightarrow \beta} & \Gamma_{\gamma \rightarrow \alpha}+\Gamma_{\gamma \rightarrow \beta}
\end{smallmatrix}
\right).
\label{eq:Gamma_mat}
\end{equation}
The stationary populations $\widetilde{\boldsymbol\rho}^{\infty}$ are obtained by solving the system $\boldsymbol\Gamma \widetilde{\boldsymbol\rho}^{\infty} = \mathbf{0}$ with $\left\| \widetilde{\boldsymbol\rho}^{\infty} \right\|_{1} = 1$.

\subsection{Dipole force}

If the Rabi frequencies $\Omega_{1}$ and $\Omega_{2}$ are functions of the spatial position coordinates $\mathbf{r}$, as expected for Gaussian optical beams, then the overall Hamiltonian of the atom, including both the internal and external degrees of freedom, is
\begin{equation}
H = \frac{\mathbf{p}^{2}}{2m} + V (\mathbf{r}),
\end{equation}
where $\mathbf{p}$ is the atomic momentum operator and $m$ is the mass of a single atom.

The dipole force operator can be defined from the Heisenberg equation of motion for the momentum operator \cite{dalibard1985}:
\begin{equation}
\mathbf{F} = \frac{d}{dt} \mathbf{p} = \frac{i}{\hbar} \left[ H,\mathbf{p} \right] = - \nabla V (\mathbf{r}).
\end{equation}
Over a time interval long compared to the lifetimes $1/\Gamma_{1}$ and $1/\Gamma_{2}$, the mean force experienced by an atom at position $\mathbf{r}$ is obtained from the average for the steady internal state,
\begin{equation}
\left\langle \mathbf{F} (\mathbf{r}) \right\rangle = \mathrm{Tr} \left[ \mathbf{F} (\mathbf{r}) \rho_{\infty} (\mathbf{r}) \right] = - \mathrm{Tr} \left[ \left( \nabla V (\mathbf{r}) \right) \rho_{\infty} (\mathbf{r}) \right].
\label{eq:opt_potential}
\end{equation}
Note that $\left\langle \mathbf{F} (\mathbf{r}) \right\rangle \neq - \nabla \left\langle V (\mathbf{r}) \right\rangle$ since the density operator $\rho_{\infty}$ depends on the position $\mathbf{r}$.

Since the trace of an operator is independent of the basis choice, the mean dipole force can be written using the steady state density matrix in the dressed states basis: $\left\langle \mathbf{F} \right\rangle = \mathrm{Tr} \left( \mathbf{F} \rho_{\infty} \right) = \mathrm{Tr} \left( U \mathbf{F} U^{\dagger} U \rho_{\infty} U^{\dagger} \right) = \mathrm{Tr} ( \widetilde{\mathbf{F}} \widetilde{\rho}_{\infty} ) = \langle \widetilde{\mathbf{F}}\rangle$. In that case, a dipole force can be associated with each dressed state,
\begin{equation}
\left\langle \widetilde{\mathbf{F}} \right\rangle = \left\langle \widetilde{\mathbf{F}}_{\alpha} \right\rangle + \left\langle \widetilde{\mathbf{F}}_{\beta} \right\rangle + \left\langle \widetilde{\mathbf{F}}_{\gamma} \right\rangle,
\label{eq:mean_dip_force}
\end{equation}
where $\left\langle \widetilde{\mathbf{F}}_{k} \right\rangle = - \widetilde{\rho}^{\infty}_{kk} \nabla \widetilde{E}_{k}$. The potential can then be obtained by integrating the force.

\section{\label{sec:heating_rates}Heating rates}

As we will see in Sec.~\ref{sec:possible_config}, the doubly driven $\Xi$-system can create a double barrier potential in exchange for a non-negligible spontaneous emission rate near the barrier position. Those spontaneous emission events are a source of heating. It is thus necessary to quantify the various heating processes to estimate whether the proposed method is viable.

The stochastic fluctuations of the atomic momentum result in a random walk of the trajectory in momentum space that increases the kinetic energy of an atom in the guide. This phenomena limits the atomic lifetime and thus the guided distance for a given mean velocity of the atoms. Two stochastic processes are responsible for this effect:
\begin{itemize}
\item the radiation pressure due to the discrete nature of the spontaneous scattering events and the random direction in which the photon is emitted;
\item the dipole force fluctuation due to the random variation of the internal atomic state \cite{gordon1980,dalibard1985}.
\end{itemize}

The velocity dispersion induced by a given process is characterized by a diffusion coefficient $D$. It is defined so that an infinitesimal variation of the velocity during a time interval $dt$ satisfies
\begin{equation}
d\mathbf{v} = \sqrt{D} d\mathbf{W}_{t},
\label{eq:wiener_inc_vel}
\end{equation}
where $d\mathbf{W}_{t} = \boldsymbol{\zeta} \sqrt{dt}$ is the increment of a Wiener noise process, and $\boldsymbol{\zeta}$ is a random vector with a normally distributed norm of unitary variance ($\left| \boldsymbol{\zeta} \right| \in \mathcal{N}(0,1)$) and a random direction.

It must be noted that all the heating processes result from the same stochastic source: the spontaneous emission. As a consequence, the processes are not only correlated but synchronized since the change of the recoil and dipolar forces occurs simultaneously. This effect may modify the resulting diffusion coefficient compared to the unsynchronized case where the two processes are uncorrelated. Nevertheless, in the following we assume that the processes are independent and do not consider any synchronization effect.

We will now determine the diffusion coefficients associated with both the radiation pressure and the dipole force fluctuation.

\subsection{Radiation pressure fluctuation}

The diffusion coefficient associated with the radiation pressure at the wavelength $\lambda_{k}$ ($k=1,2$) is
\begin{equation}
D_{k} (\mathbf{r}) = \left( \frac{\hbar k_{R}}{m} \right)^{2} \Gamma_{S}^{(k)} (\mathbf{r}),
\end{equation}
where $k_{R}$ is the single photon recoil momentum at $\lambda_{k}$, $m$ is the atomic mass, and the scattering rate is proportional to the decay rate and to the population in the bare state $\left| k \right\rangle$: $\Gamma_{S}^{(k)}(\mathbf{r}) = \Gamma_{k} \rho^{\infty}_{kk}(\mathbf{r})$. 

If the diffusion coefficient in a given spatial direction is desired, a prefactor must be applied depending on the polarization of the considered transition. For example, for the isotropic diffusion due to $s$-wave scattering the prefactor is $1/3$ for all directions, whereas for a $\sigma^{+/-}$ transition a factor of $1/4$ must be applied for the two directions orthogonal to the dipole axis and a factor of $1/2$ for the last direction.

\subsection{\label{sec:dip_force_fluc}Dipole force fluctuation}

Since the internal atomic state changes at random times, the dipole force fluctuates around its mean value [Eq.~(\ref{eq:mean_dip_force})]. From the rate equations of the dressed states populations [Eq.~(\ref{eq:rate_eqns})], we see that the stochastic evolution of the dipole force $\{ \widetilde{\mathbf{F}}_{t}, t \geq 0 \}$ is a continuous time Markovian process with infinitesimal generator $\mathbf{\Gamma}$ given by Eq.~(\ref{eq:Gamma_mat}).

The diffusion coefficient is given by Ref.~\cite{dalibard1985}:
\begin{equation}
D_{\rm dip} = \lim_{t \rightarrow \infty} \frac{1}{m^{2}} \int_{0}^{\infty} \left( \left\langle \widetilde{\mathbf{F}}_{t} \widetilde{\mathbf{F}}_{t+\tau} \right\rangle - \left\langle \widetilde{\mathbf{F}}_{t} \right\rangle^{2} \right) \; d\tau.
\label{eq:diff_coeff_def}
\end{equation}
Since the process is Markovian, we show in App.~\ref{app:dip_force_fluc} that the diffusion coefficient along a direction $\varepsilon \in \{ x,y,z \}$ satisfies the following relation:
\begin{equation}
D_{\rm dip}^{(\varepsilon)} = \frac{1}{m^{2}} \left\langle \widetilde{\mathbf{F}}^{(\varepsilon)}, \boldsymbol\Gamma^{\sharp} \widetilde{\mathbf{F}}^{(\varepsilon)} \right\rangle,
\end{equation}
where $\widetilde{\mathbf{F}}^{(\varepsilon)} = ( \widetilde{\mathbf{F}}^{(\varepsilon)}_{\alpha},\widetilde{\mathbf{F}}^{(\varepsilon)}_{\beta},\widetilde{\mathbf{F}}^{(\varepsilon)}_{\gamma} )^{T}$ is the vector of the forces in the different dressed states with $\widetilde{\mathbf{F}}^{(\varepsilon)}_{k} = - (\nabla \widetilde{E}_{k} )_{\varepsilon}$, $\mathbf{\Gamma}^{\sharp}$ is the group inverse of $\mathbf{\Gamma}$, and the scalar product is defined as $\left\langle \mathbf{x}, \mathbf{y} \right\rangle = \sum_{i} \widetilde{\rho}^{\infty}_{ii} x_{i} y_{i}$. This formula allows us to calculate the diffusion coefficient of any system described by a set of rate equations (Markovian process). Moreover, it provides an efficient way to perform a numerical estimation.

\section{Implementation with Rubidium 87 atoms}

We now present how to implement the guided continuous source with $^{87}$Rb atoms. We first show how an effective $\Xi$-system can be realized within the complicated structure of $^{87}$Rb transitions using a proper choice of polarizations for the optical fields. We then use the previous results to determine the parameters for a suitable configuration and analyze the atom loading process into the guide based on the atom-diode effect. We also explain how the transport velocity along the guide can be controlled using counter-propagating beams. Finally, we simulate atom trajectories along the guide using stochastic Langevin-like differential equations to quantify the expected performance in terms of guided efficiency along a given distance and output velocity.

\subsection{\label{sec:doubly_closed_trans}Doubly-closed transitions}

A closed $\Xi$-system within the $^{87}$Rb structure is realized using the $5^{2} \mathrm{S}_{1/2} \rightarrow 5^{2} \mathrm{P}_{3/2}$ transition at 780.24~nm and the $5^{2} \mathrm{P}_{3/2} \rightarrow 4^{2} \mathrm{D}_{5/2}$ transition at 1529.37~nm \cite{lee2007}. Therefore, it realizes a very good approximation of the theoretical model developed above. Note that this idea can be generalized to other alkali-metal atoms since they have analogous structures. Furthermore, the choice of the 1529~nm wavelength is appealing since laser technologies developed for optical telecommunications can be used.

\begin{figure}[!h]
\begin{center}
\includegraphics[width=8.5cm,keepaspectratio]{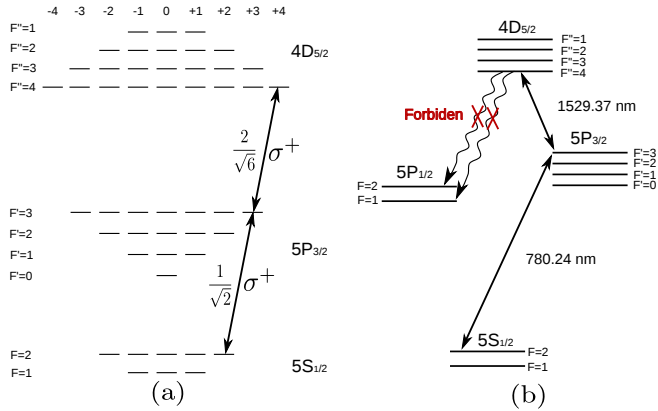}
\caption{(Color online) Cycling transitions in $^{87}$Rb. (a) The use of $\sigma^{+}$ polarized light allows us to drive a doubly closed transition while providing optical pumping. (b) The decay from the $4^{2} \mathrm{D}_{5/2}$ state to the $5^{2} \mathrm{P}_{1/2}$ state is forbidden.}
\label{fig:cycling_trans}
\end{center}
\end{figure} 

The use of $\sigma^{+}$ polarized electromagnetic fields both optically pumps the sample and drives closed transitions between the Zeeman substates $\left|F=2,m_{F}=2 \right\rangle \rightarrow \left|F'=3,m'_{F}=3 \right\rangle$ and $\left| F'=3,m'_{F}=3 \right\rangle \rightarrow \left|F''=4,m''_{F}=4 \right\rangle$, as depicted in Fig.~\ref{fig:cycling_trans}~(a). Note also that, with this choice of hyperfine states, a decay into the $5^{2} \mathrm{P}_{1/2}$ state is forbidden ($\Delta F > 1$), as shown in Fig.~\ref{fig:cycling_trans}~(b). Such a configuration is thus equivalent to the $\Xi$ system sought for. Moreover, since an atom in $\left| F=2 \right\rangle$ loaded inside the guide cannot decay into $\left| F=1 \right\rangle$ after a spontaneous emission event the atom will remain in the guided state, thus avoiding its loss during propagation. 

\begin{figure}[!h]
\begin{center}
\includegraphics[width=8.5cm,keepaspectratio]{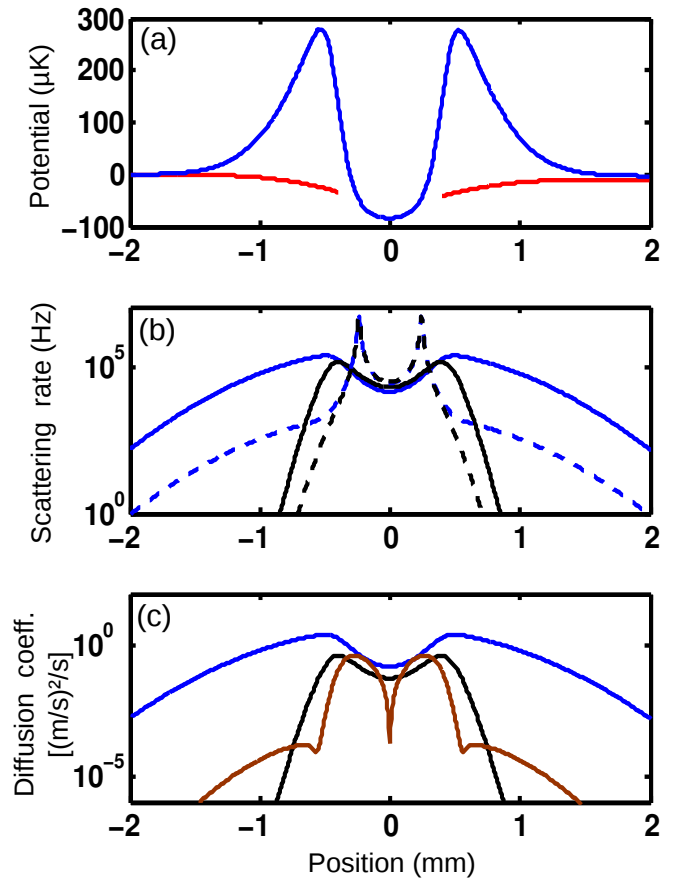}
\caption{(Color online) Dipole trap properties vs transverse position at the waist of the beams. (a) Dipole potential for an atom in $\left|F=2 \right\rangle$ (blue) and $\left|F=1 \right\rangle$ (red). The potential in $\left|F=1 \right\rangle$ is represented only in a region where the scattering rate is low since $\left|F=1 \right\rangle$ is not a stationary state of the system. (b) Scattering rates at 780~nm (blue) and 1529~nm (black) for an atom in $\left|F=2 \right\rangle$ (solid line) and $\left|F=1 \right\rangle$ (dashed line). (c) Diffusion coefficient for an atom in $\left|F=2 \right\rangle$ due to the recoil at 780~nm (blue), at 1529~nm (black) and to the fluctuations of the dipole force (brown).}
\label{fig:internal_state_param}
\end{center}
\end{figure}

\subsection{\label{sec:possible_config}A possible configuration}

We analyze the properties of the guide obtained in the following specific configuration. The 780 nm beam has a waist of 1~mm and an intensity of $1000$~mW/cm$^{2}$ (the optical power is thus 16~mW), blue detuned by $\Delta_{1} = - 2 \pi \times 485$~MHz from the $\left| F=2 \right\rangle \rightarrow \left| F'=3 \right\rangle$ transition and $\Gamma_{1} = 2 \pi \times 6.065$~MHz \cite{steck}.

The beam at 1529 nm has a waist of 300~$\mu$m and thus a Rayleigh range of 18~cm. Its intensity is $3 \times 10^{5}$~mW/cm$^{2}$ (a power of 420~mW), and the detuning is $\Delta_{2} = + 2 \pi \times 133$~MHz to the red of the $\left| F'=3 \right\rangle \rightarrow \left| F''=4 \right\rangle$ transition. The transition linewidth is $\Gamma_{2} = 2 \pi \times 2$~MHz \cite{heavens1961,moon2007}.

From these parameters, we determine the steady state density matrix versus spatial position using the results in App.~\ref{app:steady_rho}, and estimate the dipole potential after integrating the dipole force [Eq.~(\ref{eq:mean_dip_force})]. The results are presented in Fig.~\ref{fig:internal_state_param}. 


\subsection{Loading of the guide}

We see in Fig.~\ref{fig:internal_state_param}~(a) that, for an atom in $\left| F=2 \right\rangle$, the potential has a double barrier shape, creating a pipe-like potential due to the cylindrical symmetry of the problem. The guide depth is 365~$\mu$K. Conversely, for an atom in $\left| F=1 \right\rangle$, the potential is flat at the barrier position. The proposed configuration thus implements the state-dependent potential necessary to obtain the atom-diode effect.

Nevertheless, for the atom-diode effect to occur an atom in $\left| F=1 \right\rangle$ must first pass through the barrier and then be repumped into $\left| F=2 \right\rangle$, after passing the barrier maximum. In other words, an atom in $\left| F=1 \right\rangle$ should not scatter more than a few photons between the moment when it enters the 780~nm beam and the moment when it reaches the position of the barrier maximum. From Fig.~\ref{fig:internal_state_param}~(b), we see that the barrier width is about 0.1~mm wide and the scattering rate is below 1~kHz before the barrier maximum. As a consequence, an atom with a velocity of $\sim$~10~cm/s scatters less than one photon while crossing the barrier, resulting in a low probability for the atom to be repumped. Conversely, the scattering rate of an atom in $\left| F=1 \right\rangle$ strongly increases at the guide center and the atom will be pumped into the trap sensitive state, thus realizing the atom-diode based loading. This effect is studied in details using a numerical simulation in Sec.~\ref{sec:sim_guide_loading}.


\subsection{Heating rates}

By creating a differential light-shift, the 1529~nm radiation not only creates the guide barriers but also reduces the scattering rates, and thus the diffusion coefficients due to photon recoil, at 780~nm and 1529~nm by about two orders of magnitude [Fig.~\ref{fig:internal_state_param}~(c)]. This effect strongly lowers the heating rate of the guided sample. 

Nevertheless, the heating rate due to the recoils cannot be made arbitrarily small. The reason is that the diffusion due to the dipole force fluctuation increases with the intensity at 1529~nm because of the enhanced light-shift. A compromise must thus be found between the spontaneous emission rate and the dipole force diffusion. In the present configuration, the parameters are adjusted so that the diffusion due to the recoil is of the same order as the one induced by the dipole force fluctuation.

\subsection{\label{sec:transport}Transport along the guide}

To transport the atoms, an extra force must be applied along the guide axis. By counterpropagating the beams at 780~nm and at 1529~nm, the scattering imbalance between the two beams pushes the atoms in a given direction with a force
\begin{equation}
\mathbf{F}_{\rm push} ( \mathbf{r}) = \hbar \partial_{t} \mathbf{k} = \hbar \left( \Gamma_{S}^{(1)} ( \mathbf{r}) k_{R}^{(1)} - \Gamma_{S}^{(2)} ( \mathbf{r}) k_{R}^{(2)} \right) \mathbf{u}_{z},
\end{equation}
where $k_{R}^{(k)}$ is the single-photon recoil impulsion at the wavelength $\lambda_{k}$ and $\mathbf{u}_{z}$ is the unit vector in the $z$ direction. The spatial distribution of the pushing force is depicted in Fig.~\ref{fig:push_force}. The velocity of an atom inside the guide is almost constant whereas the atom accelerates near the barrier.

\begin{figure}[!h]
\begin{center}
\includegraphics[width=7.5cm,keepaspectratio]{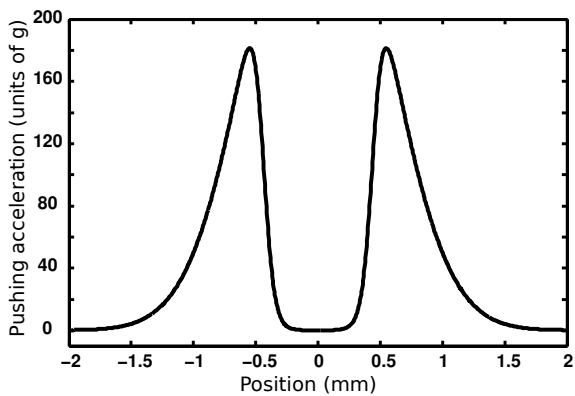}
\caption{Pushing acceleration vs the transverse position resulting from the counterpropagation of the 780 and 1529~nm beams ($g=9.81$~m/s$^{2}$).}
\label{fig:push_force}
\end{center}
\end{figure}


\subsection{Numerical simulation}

\begin{figure*}
\includegraphics[width=18cm,keepaspectratio]{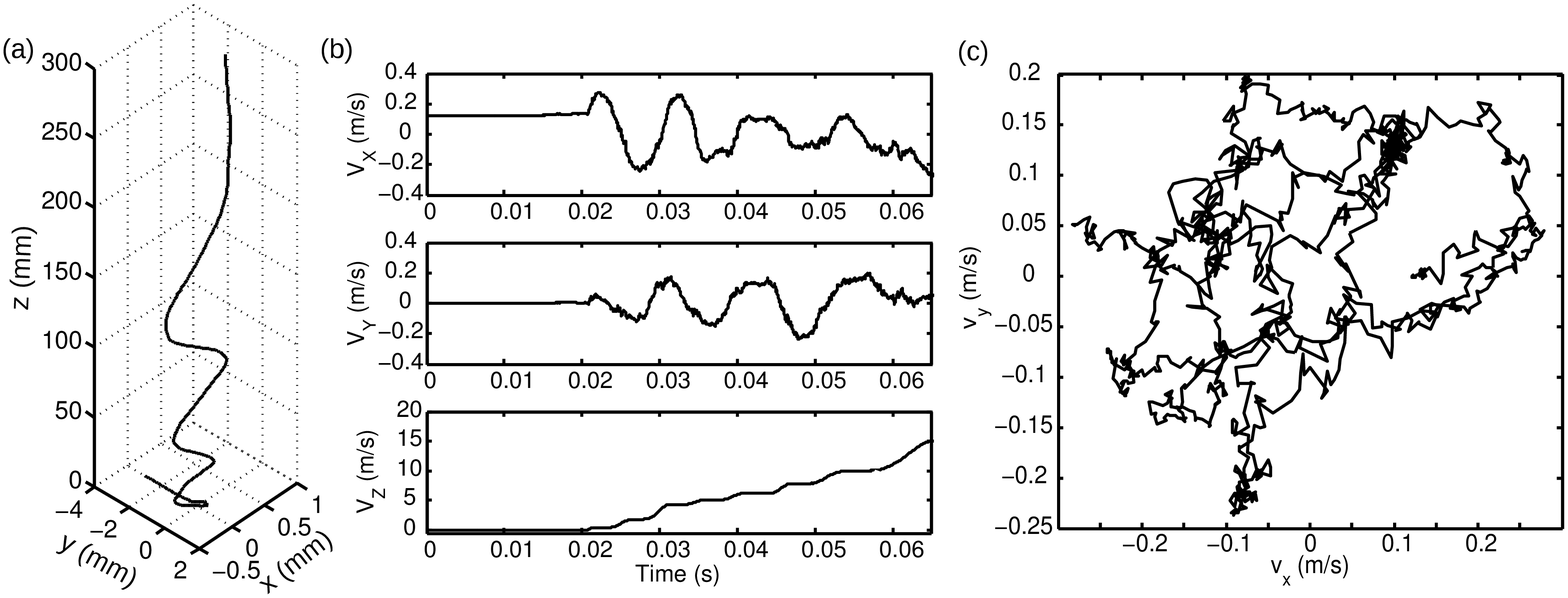}
\caption{Trajectory of an atom along the guide obtained by numerical integration of the system of stochastic differential equations [Eqs.~(\ref{eq:sde_speed})]. The simulation includes the loading of the atom inside the guide. The atom is initially in $|F=1 \rangle$ and placed at 1.5 mm from the guide axis. It is launched towards the guide with a velocity distributed according to a Maxwell-Boltzmann law with temperature $T_{0} = 50 \; \mu$K. (a) Trajectory along the guide. The atom is reflected off the guide barrier until it escapes the guide when the transverse kinetic energy is higher than the height of the barrier. (b) Evolution of the velocities along the three directions. When a reflection occurs, the transverse velocities exhibit sudden sign changes and the longitudinal velocity increases due to the higher spontaneous emission rate near the barrier. (c) The Brownian motion of the transverse velocities responsible for the heating of the sample by increasing the radial kinetic energy.}
\label{fig:trajectory}
\end{figure*}

Since the heating processes limit the lifetime of the guided sample, a compromise must be found between the transport velocity and the guided length: the atomic sample at the guide output should be cold -- meaning that the output velocity and its dispersion should be of the order of a few meters per second -- while the sample should be guided along a few centimeters. This distance is sufficient, for instance, to deliver atoms inside the magnetic shield required for any kind of high-performance atom interferometer, particularly since the MOT magnetic field is continuously operated.

To validate that the proposed configuration satisfies these requirements we numerically simulate the trajectories of atoms along the guide, solving a set of Langevin equations. The statistical analysis of the behavior over a large number of trajectories allows us to estimate the expected performances of the system in terms of guided distance and output velocity distribution. Moreover, the loading process is included in the simulation to consider its influence on the guided sample.

\subsubsection{Langevin equations}

The simulation relies on the fact that the internal degrees of freedom of an atom evolve on a much shorter time-scale (limited by the lifetime of the excited states to submicroseconds) than the external ones ($\sim$~ms). This allows us to employ a semi-classical description of the problem, describing the internal state evolution according to quantum mechanics while modeling the external dynamics with Langevin-like equations (Newton equations with stochastic velocity increments). 

The momentum diffusion of the atoms generates a random walk of their velocities in the transverse direction, inducing stochastic heating. This random walk can be modelled from the diffusion coefficients derived in Sec.~\ref{sec:heating_rates} and with a Wiener increment [Eq.~(\ref{eq:wiener_inc_vel})] in the differential equation of motion. Explicitly, the velocity vector $\mathbf{v} = \left( v_{x}, v_{y}, v_{z} \right)$ evolution is governed by the following system of stochastic differential equations:
\begin{subequations}
\label{eq:sde_speed}
\begin{eqnarray}
d v_{x}(t) & = & \frac{1}{m} F_{x} \left( \mathbf{r} (t) \right) dt + \sqrt{D_{1}\left( \mathbf{r} (t) \right)} dW^{(1)}_{x,t} \nonumber \\
& & + \sqrt{D_{2}\left( \mathbf{r} (t) \right)} dW^{(2)}_{x,t} + \sqrt{D_{\rm dip}^{(x)}\left( \mathbf{r} (t) \right)} dW^{(\mathrm{dip})}_{x,t}, \nonumber \\ && \\
d v_{y}(t) & = & \frac{1}{m} F_{y} \left( \mathbf{r} (t) \right) dt + \sqrt{D_{1}\left( \mathbf{r} (t) \right)} dW^{(1)}_{y,t} \nonumber \\
& & + \sqrt{D_{2}\left( \mathbf{r} (t) \right)} dW^{(2)}_{y,t} + \sqrt{D_{\rm dip}^{(y)}\left( \mathbf{r} (t) \right)} dW^{(\mathrm{dip})}_{y,t}, \nonumber \\ && \\
d v_{z}(t) & = & \frac{1}{m} \left[ F_{z} \left( \mathbf{r} (t) \right) + F_{\rm push} \left( \mathbf{r} (t) \right) \right] dt \nonumber \\
& & + \sqrt{D_{1}\left( \mathbf{r} (t) \right)} dW^{(1)}_{z,t} + \sqrt{D_{2}\left( \mathbf{r} (t) \right)} dW^{(2)}_{z,t} \nonumber \\
&& + \sqrt{D_{\rm dip}^{(z)}\left( \mathbf{r} (t) \right)} dW^{(\mathrm{dip})}_{z,t}. 
\end{eqnarray}
\end{subequations}

We solve this system using a Monte-Carlo method to simulate atom trajectories along the guide and employ the Runge-Kutta method \cite{sauer2012} to perform the numerical integration of the system.

\subsubsection{\label{sec:sim_guide_loading}Initial conditions and loading of the guide}

We first describe the initial conditions of the simulation and use the simulation to analyze the loading efficiency -- that is, the probability for an atom to enter the guide -- and the velocity distribution of the loaded atoms.

The atom is initially in the $| F=1 \rangle$ state and located along the $y$ axis at a distance of 3~mm from the guide axis. Since the atoms in the MOT are far from degeneracy, the initial velocity $v_{0}$ is distributed according to a Maxwell-Boltzmann law
\begin{equation}
f \left( v \right) = \sqrt{\frac{2}{\pi}} \frac{v^{2}}{\delta v^{3}} \exp \left(- \frac{v^{2}}{2 \delta v^{2}} \right),
\label{eq:maxwell-boltzmann}
\end{equation}
for a temperature $T_{0} = 50 \; \mu$K corresponding to a velocity dispersion $\delta v = \sqrt{k_{B} T_{0}/m} \sim 7$~cm/s, and the initial velocity vector is $\mathbf{v}_{0} = (v_{0},0,0)$. Moreover, the Doppler frequency shift ($\Delta \nu_{D} = -v/\lambda$) is included in the simulation.

During the propagation of the atom towards the guide, we estimate at each step of the simulation the probability for the atom to be pumped into the stable state $|F=2, m_{F} = +2 \rangle$ and then randomly determine whether the atom is pumped or not. More precisely, we use the knowledge of the scattering rate at each position $\Gamma_{S}^{(1)} (\mathbf{r})$ and the fact that on average 2.4 photons must be scattered to pump the atom. We consider that an atom is loaded inside the guide when it is pumped at less than 0.5~mm from the guide center (which corresponds to the guide radius defined from the position of the barrier maximum). From the simulation of $10^{4}$ stochastic trajectories, we find that about 21\% of the atoms are loaded into the guide.

After selecting the trajectories that lead to loaded atoms, we analyze the velocity distribution of the atoms before their further transport along the guide. We define the loaded velocity as that when the atom first reaches the plane $y=0$. From a large number of loading simulations we estimate the velocity distribution and the result is presented in Fig.~\ref {fig:loading_velocities}. The resulting distribution is well described by the initial Maxwell-Boltzmann distribution [Eq.~(\ref{eq:maxwell-boltzmann})] for $T_{0}=50 \; \mu$K shifted by 2.3~cm/s, and the mean velocity of the loaded atoms is 13.4~cm/s. This deterministic increase of kinetic energy has two contributions: first the $| F=1 \rangle$ potential is slightly attractive, and secondly the atom rolls down the barrier potential after being pumped into $| F=2 \rangle$. However, it is interesting to note that the velocity dispersion of the loaded atoms remains that of the atoms in the MOT.

\begin{figure}[!h]
\begin{center}
\includegraphics[width=7.5cm,keepaspectratio]{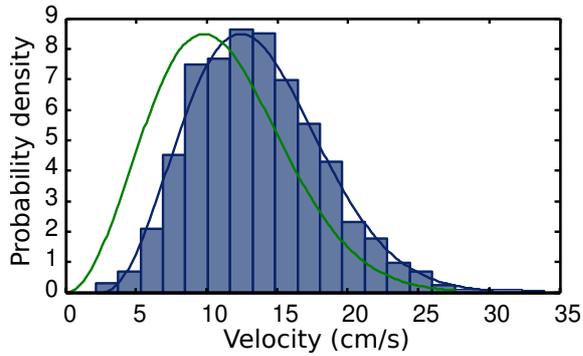}
\caption{(Color online) Velocity distribution of the loaded atoms. The histogram is the distribution obtained from the simulation of 2132 Monte-Carlo trajectories, the solid green line is the initial Maxwell-Boltzmann distribution with $T_{0} =50 \; \mu$K and the solid blue line is the initial distribution shifted by 2.3~cm/s.}
\label{fig:loading_velocities}
\end{center}
\end{figure}

\subsubsection{Transport along the guide}

We now analyze the transport of the atoms along the guide. An example of a simulated trajectory is presented in Fig.~\ref{fig:trajectory}. We see that the atom is guided along the beam axis by reflections off the potential barrier. During the reflections, the transverse velocities reverse while the longitudinal velocity increases due to the higher scattering rate near the guide barrier. This demonstrates the effect of the pushing force (discussed in Sec.~\ref{sec:transport}), which is nondeterministic since it depends on the stochastic atomic trajectory. 

\begin{figure}[!h]
\begin{center}
\includegraphics[width=8.5cm,keepaspectratio]{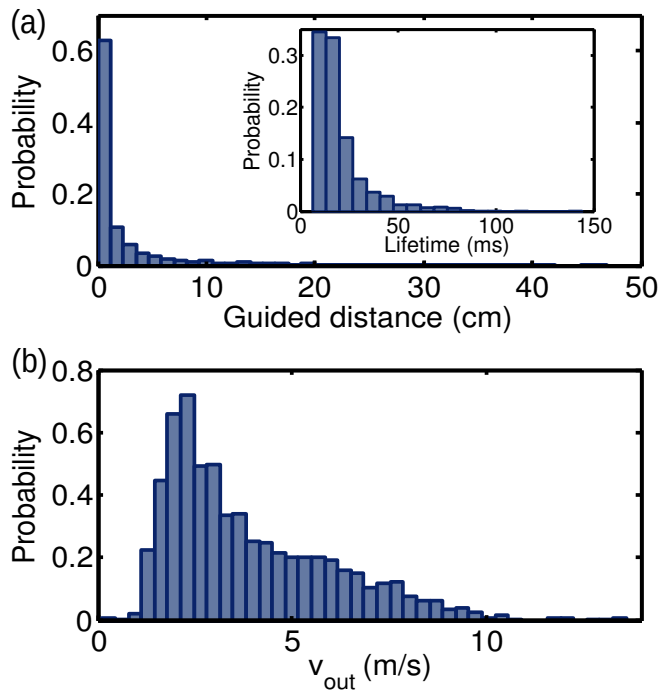}
\caption{(Color online) Statistical properties obtained from the simulation of 2132 atomic trajectories. (a) Probability for an atom to be guided over a given distance. Inset: The distribution for the lifetime of an atom in the guide. (b) Longitudinal velocity distribution at a guided distance of 5~cm.}
\label{fig:stats_sim_num}
\end{center}
\end{figure}

Simulating a large number of atomic trajectories allows us to estimate the guide performance. The distribution of the guided distances -- an atom is lost when it is further than 0.5~mm from the guide axis -- shows that a significant fraction of the loaded atoms ($\sim$~16~\%) have been guided over at least 5~cm [Fig.~\ref{fig:stats_sim_num}~(a)]. Moreover, from the velocity distribution at the guided distance of 5~cm [Fig.~\ref{fig:stats_sim_num}~(b)] we determine that the mean velocity is 3.9~m/s with a dispersion of 2.1~m/s, corresponding to a temperature of about 25~mK. Note that transverse cooling could be added \cite{aghajani-talesh2010}: either to maintain a given propagation distance while reducing the longitudinal velocity, or to increase the propagation distance for a given velocity. A better control of the longitudinal velocity could be achieved by retropropagating the beams at 780~nm and 1529~nm and precisely tuning the power ratios between the forward and backward directions; in that case, a grating would be generated and sub-Doppler cooling mechanisms could be exploited.

 
\section{Conclusion}

We proposed and theoretically analyzed an all-optical method to produce a continuous source of cold atoms in a spin-polarizing guide. The resulting system is a combination of several physical effects, and notably an atom-diode effect based on light-shift engineering in a three-level atom. 

To engineer the system, we modeled the dipole forces in a doubly-driven $\Xi$-system and provided a deep description of the diffusion processes responsible for the heating of the atomic sample. These results constitute a general theoretical framework which can be used to design other atom control methods based on $\Xi$-systems.

We finally presented and studied in detail a possible implementation with $^{87}$Rb. In particular, we isolated a closed three-level transition which can be driven with reliable optical sources based on diode laser and telecommunication technologies. Since the guide consists in overlapping two optical beams, the experimental implementation should be compact and simple to operate which are interesting features for embedded applications. We showed with numerical simulations that a guiding distance of several centimeters is achievable and compatible with slow velocities of a few meters per second.

\section*{Acknowledgments}

We thank S.~Palacios and N.~Martinez de Escobar for comments. This work was supported by the Spanish MINECO project MAGO (FIS2011-23520), by the European Research Council project AQUMET and by Fundació Privada Cellex.

\appendix

\section{Optical Bloch equations and steady state solution}

\subsection{\label{app:opt_bloch_eqs}Optical Bloch equations}

In the $\left\{ \left| 1 \right\rangle,\left| 2 \right\rangle,\left| 3 \right\rangle \right\}$ basis, the density operator can be written as
\begin{equation}
\rho = \left( 
\begin{array}{ccc}
\rho_{11}     & \rho_{12}     & \rho_{13} \\
\rho_{12}^{*} & \rho_{22}     & \rho_{23} \\
\rho_{13}^{*} & \rho_{23}^{*} & \rho_{33}
\end{array}
\right),
\end{equation}
and Eq.~(\ref{eq:master_eq}) provides the optical Bloch equations for a $\Xi$-system:
\begin{subequations}
\begin{eqnarray}
\partial_{t} \rho_{11} & = & i \frac{\Omega_{1}}{2} \left( \rho_{12} - \rho_{12}^{*} \right) + \Gamma_{1} \rho_{22}, \\
\partial_{t} \rho_{22} & = & -i \frac{\Omega_{1}}{2} \left( \rho_{12} - \rho_{12}^{*} \right) + i \frac{\Omega_{2}}{2} \left( \rho_{23} - \rho_{23}^{*} \right) \nonumber \\
& & \;\;\;\;\; + \Gamma_{2} \rho_{33} - \Gamma_{1} \rho_{22}, \\
\partial_{t} \rho_{33} & = & -i \frac{\Omega_{2}}{2} \left( \rho_{23} - \rho_{23}^{*} \right) - \Gamma_{2} \rho_{33}, \\
\partial_{t} \rho_{12} & = & \left( - \frac{\Gamma_{1}}{2} + i \Delta_{1} \right) \rho_{12} - i \frac{\Omega_{1}}{2} \left( \rho_{22} - \rho_{11} \right) \nonumber \\
& & \;\;\;\;\; + i \frac{\Omega_{2}}{2} \rho_{13}, \\
\partial_{t} \rho_{13} & = & \left( - \frac{\Gamma_{2}}{2} + i \Delta_{2} \right) \rho_{13} + i \frac{\Omega_{2}}{2} \rho_{12} - i \frac{\Omega_{1}}{2} \rho_{23},  \nonumber \\ && \\
\partial_{t} \rho_{23} & = & \left[ - \frac{1}{2} \left( \Gamma_{1} + \Gamma_{2} \right) + i \left( \Delta_{2} - \Delta_{1} \right) \right] \rho_{23} \nonumber \\
& & \;\;\;\;\; - i \frac{\Omega_{2}}{2} \left( \rho_{33} - \rho_{22} \right) -i \frac{\Omega_{1}}{2} \rho_{13}.
\end{eqnarray}
\end{subequations}

\subsection{\label{app:steady_rho}Steady state of the density matrix}

The steady state solution of the optical Bloch equations is obtained for $\partial_{t} \rho_{\infty} = 0$. Since $\mathrm{Tr} (\rho) = 1$, one has the following relation between the populations: $\rho_{11}+\rho_{22}+\rho_{33}=1$. Moreover, $\rho_{ij}-\rho_{ij}^{*}=2i \mathrm{Im} \rho_{ij}$ and the steady state thus obeys the following system of equations:
\begin{subequations}
\begin{eqnarray}
&- \Omega_{1} \mathrm{Im} \rho_{12} + \Gamma_{1} \rho_{22} = 0, \label{eq:1} \\
&\Omega_{2} \mathrm{Im} \rho_{23} - \Gamma_{2} \rho_{33} = 0, \label{eq:2} \\
&\left( - \frac{\Gamma_{1}}{2} + i \Delta_{1} \right) \rho_{12} - i \frac{\Omega_{1}}{2} \left( 2\rho_{22} + \rho_{33}-1 \right) \nonumber \\ 
& \;\;\;\; + i \frac{\Omega_{2}}{2} \rho_{13} = 0, \label{eq:3} \\
&\left( - \frac{\Gamma_{2}}{2} + i \Delta_{2} \right) \rho_{13} + i \frac{\Omega_{2}}{2} \rho_{12} - i \frac{\Omega_{1}}{2} \rho_{23} = 0, \label{eq:4} \\
&\left[ - \frac{1}{2} \left( \Gamma_{1} + \Gamma_{2} \right) + i \left( \Delta_{2} - \Delta_{1} \right) \right] \rho_{23} \nonumber \\
& \;\;\;\; - i \frac{\Omega_{2}}{2} \left( \rho_{33} - \rho_{22} \right) -i \frac{\Omega_{1}}{2} \rho_{13} = 0. \label{eq:5}
\end{eqnarray}
\end{subequations}
Note that the second equation has been removed since it can be obtained from a linear combination of the first and the third equations.

To solve this system, we split the density matrix components into real and imaginary parts: $\rho_{ij}=\rho_{ij}^{R}+i\rho_{ij}^{I}$. From Eqs.~(\ref{eq:1}) and (\ref{eq:2}) we obtain the relationship between the coherences and the populations in $\left| 2 \right\rangle$ and $\left| 3 \right\rangle$:
\begin{subequations}
\label{eq:populations_to_imag}
\begin{eqnarray}
\rho_{22} &=& \frac{\Omega_{1}}{\Gamma_{1}} \rho_{12}^{I}, \\
\rho_{33} &=& \frac{\Omega_{2}}{\Gamma_{2}} \rho_{23}^{I}.
\end{eqnarray}
\end{subequations}
The real parts of Eqs.~(\ref{eq:3}), (\ref{eq:4}) and (\ref{eq:5}) provide the following relations between the real and imaginary parts of the coherence terms:
\begin{equation}
\left( 
\begin{array}{c}
\rho_{12}^{R} \\
\rho_{13}^{R} \\
\rho_{23}^{R}
\end{array}
\right) = 
\left( 
\begin{array}{ccc}
-2\Delta_{1}/\Gamma_{1} & -\Omega_{2}/\Gamma_{1} & 0 \\
-\Omega_{2}/\Gamma_{2} & -2\Delta_{2}/\Gamma_{2} & \Omega_{1}/\Gamma_{2} \\
0 & \frac{\Omega_{1}}{\Gamma_{1}+\Gamma_{2}} & - 2 \frac{\Delta_{2}-\Delta_{1}}{\Gamma_{1}+\Gamma_{2}}
\end{array}
\right)
\left( 
\begin{array}{c}
\rho_{12}^{I} \\
\rho_{13}^{I} \\
\rho_{23}^{I}
\end{array}
\right).
\label{eq:rho_imag_to_rho_real}
\end{equation}
From this relation and the imaginary parts of Eqs.~(\ref{eq:3}), (\ref{eq:4}), and (\ref{eq:5}), we show that the imaginary parts of the coherences are solutions of the following system:
\begin{equation}
\mathbf{M}\left( 
\begin{array}{c}
\rho_{12}^{I} \\
\rho_{13}^{I} \\
\rho_{23}^{I}
\end{array}
\right)
= \left( 
\begin{array}{ccc}
\Gamma_{1} \Omega_{1} \\
0 \\
0
\end{array}
\right), \\
\label{eq:syst_rho_imag}
\end{equation}
where
\begin{widetext}
\begin{equation}
\mathbf{M} =
\left( 
\begin{array}{ccc}
\Gamma_{1}^{2} + 4 \Delta_{1}^{2} + 2 \Omega_{1}^{2} + \frac{\Gamma_{1}}{\Gamma_{2}} \Omega_{2}^{2} & 2 \Omega_{2} \left( \Delta_{1} + \frac{\Gamma_{1}}{\Gamma_{2}} \Delta_{2} \right) & 0 \\
2 \Omega_{2} \left( \Delta_{2} + \frac{\Gamma_{2}}{\Gamma_{1}} \Delta_{1} \right) &  \Gamma_{2}^{2} + 4 \Delta_{2}^{2} + \frac{\Gamma_{2}}{\Gamma_{1}} \Omega_{2}^{2} + \frac{\Gamma_{2}}{\Gamma_{1}+\Gamma_{2}} \Omega_{1}^{2} & -2 \frac{\Gamma_{2}}{\Gamma_{1}+\Gamma_{2}} \Omega_{1} \left( \Delta_{2} - \Delta_{1} \right)  \\
\Omega_{1} \Omega_{2} \left( 1 + \frac{\Gamma_{2}}{\Gamma_{1}} \right) & 2 \Omega_{1} \left( \Delta_{2} + \frac{\Gamma_{2}}{\Gamma_{1}+\Gamma_{2}} \left( \Delta_{2}-\Delta_{1} \right) \right) & - \Omega_{1}^{2} - \Omega_{2}^{2} -4 \frac{\Gamma_{2}}{\Gamma_{1}+\Gamma_{2}} \left( \Delta_{2} - \Delta_{1} \right)^{2} - \Gamma_{2} \left( \Gamma_{1}+\Gamma_{2} \right)
\end{array}
\right).
\end{equation}
\end{widetext}
As a consequence, the determination of the inverse matrix $\mathbf{M}^{-1}$ allows us to solve the system (\ref{eq:syst_rho_imag}), and thus to fully determine the steady state density matrix $\rho_{\infty}$ using the relations (\ref{eq:rho_imag_to_rho_real}) and (\ref{eq:populations_to_imag}).

\section{Calculation of the diffusion coefficient of the dipole force fluctuation}

Here we calculate the diffusion coefficient associated with the dipole force fluctuation. Using the master equation in the dressed state basis we show that, in the secular approximation, the dressed state populations obey rate equations. This means that the stochastic internal state evolution is a continuous time Markov process. We then derive a general expression for the diffusion coefficient associated with a Markov process. This result is finally applied in the case of the dipole force diffusion.

\subsection{\label{app:trans_prob}Evolution of the dressed states populations}

The master equation (\ref{eq:master_eq}) can be written for the operators in the dressed state basis as
\begin{equation}
\partial_{t} \widetilde{\rho} = - \frac{i}{\hbar} \left[ \widetilde{V},\widetilde{\rho} \right] + \Gamma_{1} \mathcal{L} \left[ \widetilde{\sigma}_{12} \right] \widetilde{\rho} + \Gamma_{2} \mathcal{L} \left[ \widetilde{\sigma}_{23} \right] \widetilde{\rho}.
\label{eq:master_eq_dress_state}
\end{equation}
Following \cite{Cohen1992}, we project Eq.~(\ref{eq:master_eq_dress_state}) on a given dressed state $\left| i \right\rangle$. Since $\widetilde{V} \left| i \right\rangle = E_{i}\left| i \right\rangle$, one has $\left\langle i \right| [ \widetilde{V},\widetilde{\rho} ] \left| i \right\rangle = 0$, the only contribution to the population evolution of the dressed states is thus provided by the projection of the Lindblad terms:
\begin{equation}
\left\langle i \right| \mathcal{L} \left[ \widetilde{\sigma} \right] \widetilde{\rho} \left| i \right\rangle = \left\langle i \right| \widetilde{\sigma} \widetilde{\rho} \widetilde{\sigma}^{\dagger} \left| i \right\rangle - \frac{1}{2} \left\langle i \right| \widetilde{\sigma} \widetilde{\sigma}^{\dagger} \widetilde{\rho} \left| i \right\rangle - \frac{1}{2} \left\langle i \right| \widetilde{\rho} \widetilde{\sigma} \widetilde{\sigma}^{\dagger} \left| i \right\rangle.
\end{equation}
From the decomposition of the density matrix in the dressed states basis,
\begin{equation}
\widetilde{\rho} = \sum_{j,l} \widetilde{\rho}_{jl} \left| j \right\rangle \left\langle l \right|,
\end{equation}
we obtain
\begin{equation}
\left\langle i \right| \widetilde{\sigma} \widetilde{\rho} \widetilde{\sigma}^{\dagger} \left| i \right\rangle = \sum_{j,l} \widetilde{\rho}_{jl} \left\langle i \right| \widetilde{\sigma} \left| j \right\rangle \left\langle l \right| \widetilde{\sigma}^{\dagger} \left| i \right\rangle.
\end{equation}
We now assume that we are in the limit $\Omega_{1} \gg \Gamma_{1}$  and $\Omega_{2} \gg \Gamma_{2}$, in which case the coherences evolve fast compared to the populations. We can thus perform the secular approximation where the coherences are replaced by their mean value -- that is, for $j \neq l$, $\left\langle \widetilde{\rho}_{jl} \right\rangle \sim 0$ -- and thus obtain
\begin{equation}
\left\langle i \right| \widetilde{\sigma} \widetilde{\rho} \widetilde{\sigma}^{\dagger} \left| i \right\rangle = \sum_{j} \widetilde{\rho}_{jj} \left| \left( \widetilde{\sigma} \right)_{i j} \right|^{2}.
\end{equation}
Decomposing again the density operator in the dressed state basis, we have
\begin{eqnarray}
\left\langle i \right| \widetilde{\sigma} \widetilde{\sigma}^{\dagger} \widetilde{\rho} \left| i \right\rangle & = & \sum_{j, l} \widetilde{\rho}_{jl} \left\langle i \right| \widetilde{\sigma} \widetilde{\sigma}^{\dagger} \left| j \right\rangle \left\langle l \right| \left. i \right\rangle \\
& = & \sum_{j} \widetilde{\rho}_{j i} \left\langle i \right| \widetilde{\sigma} \widetilde{\sigma}^{\dagger} \left| j \right\rangle \\
& \sim & \widetilde{\rho}_{i i} \left\langle i \right| \widetilde{\sigma} \widetilde{\sigma}^{\dagger} \left| i \right\rangle
\end{eqnarray}
where the last line is obtained using the secular approximation again. From the resolution of unity $\sum_{j} \left| j \right\rangle \left\langle j \right| = \mathds{1}$, we finally obtain
\begin{eqnarray}
\left\langle i \right| \widetilde{\sigma} \widetilde{\sigma}^{\dagger} \widetilde{\rho} \left| i \right\rangle & = & \widetilde{\rho}_{i i} \left\langle i \right| \widetilde{\sigma} \left( \sum_{j} \left| j \right\rangle \left\langle j \right| \right) \widetilde{\sigma}^{\dagger} \left| i \right\rangle \\
& = & \widetilde{\rho}_{i i} \sum_{j} \left| \left( \widetilde{\sigma} \right)_{i j} \right|^{2}.
\end{eqnarray}

Combining these results according to the master equation (\ref{eq:master_eq_dress_state}), the evolution of the populations finally reduces to the following rate equations:
\begin{equation}
\partial_{t} \widetilde{\rho}_{ii} = - \left( \sum_{j} \Gamma_{i \rightarrow j} \right) \widetilde{\rho}_{ii} + \sum_{j} \Gamma_{i \rightarrow j} \widetilde{\rho}_{jj},
\label{eq:pop_eq_sec_app}
\end{equation}
where the transition rates are:
\begin{equation}
\Gamma_{i \rightarrow j} = \Gamma_{1} \left| \left( \widetilde{\sigma}_{12} \right)_{ij} \right|^{2} + \Gamma_{2} \left| \left( \widetilde{\sigma}_{23} \right)_{ij} \right|^{2}.
\end{equation}

Introducing the dressed states populations vector $\widetilde{\boldsymbol\rho} = \left( \widetilde{\rho}_{\alpha \alpha}, \widetilde{\rho}_{\beta \beta}, \widetilde{\rho}_{\gamma \gamma} \right)^{T}$, the population equations (\ref{eq:pop_eq_sec_app}) can be written as
\begin{equation}
\partial_{t} \widetilde{\boldsymbol\rho} + \boldsymbol\Gamma \widetilde{\boldsymbol\rho} = \mathbf{0},
\end{equation}
where
\begin{equation}
\boldsymbol\Gamma = \left( 
\begin{smallmatrix}
\Gamma_{\alpha \rightarrow \beta} + \Gamma_{\alpha \rightarrow \gamma} & -\Gamma_{\alpha \rightarrow \beta} & -\Gamma_{\alpha \rightarrow \gamma} \\
-\Gamma_{\beta \rightarrow \alpha} & \Gamma_{\beta \rightarrow \alpha}+\Gamma_{\beta \rightarrow \gamma} & -\Gamma_{\beta \rightarrow \gamma} \\
-\Gamma_{\gamma \rightarrow \alpha} & -\Gamma_{\gamma \rightarrow \beta} & \Gamma_{\gamma \rightarrow \alpha}+\Gamma_{\gamma \rightarrow \beta}
\end{smallmatrix}
\right).
\end{equation}
It is important to note that $\det \boldsymbol\Gamma = 0$, which means that $\boldsymbol\Gamma$ is singular.

The populations thus evolve according to:
\begin{equation}
\widetilde{\boldsymbol\rho} (t) =  e^{-\boldsymbol\Gamma t} \widetilde{\boldsymbol\rho} (0).
\label{eq:res_pop_ev}
\end{equation}

This relation allows us to determine the transition probability $\mathbf{P}_{ij} \left( \tau \right) = P \left( j,\tau | i,0 \right)$, that is, the probability to be in the state $\left| j \right\rangle$ at time $t=\tau$ given that it was in state $\left| i \right\rangle$ at time $t=0$. In other words, this is the population $\widetilde{\rho}_{jj} (\tau)$ given that initially $\widetilde{\rho}_{ii} (0) = 1$. The relation (\ref{eq:res_pop_ev}) thus provides the transition probability matrix
\begin{equation}
\mathbf{P} \left( \tau \right) = e^{-\boldsymbol\Gamma \tau}.
\end{equation}
The internal state evolution is a homogeneous continuous time Markov process associated with the transition matrix $\mathbf{P} \left( \tau \right)$. The Markov process is characterized by $\boldsymbol\Gamma$ which is the infinitesimal generator of this process: it is the generator of the semi-group $\left\{ \mathbf{P} \left( \tau \right), \tau \geq 0 \right\}$ and  forward time translation is mapped onto this semi-group though the Chapman-Kolmogorov relation $\forall \tau_{2} \geq \tau_{1}, \; \mathbf{P} \left( \tau_{1}+\tau_{2} \right) = \mathbf{P} \left( \tau_{1} \right) \mathbf{P} \left( \tau_{2} \right)$.

\subsection{\label{app:dip_force_fluc}Diffusion coefficient of the dipole force fluctuation}

We have shown that the stochastic evolution of the internal state is a homogeneous continuous time Markov process with generator $\mathbf{\Gamma}$. Since for each dressed state $\left| k \right\rangle$ there corresponds a given dipole force $\widetilde{\mathbf{F}}_{k} = - \nabla \widetilde{E}_{k}$, the stochastic dipole force evolution $\{ \widetilde{\mathbf{F}}_{t}, t \geq 0 \}$ is also a homogeneous continuous time Markov process with generator $\mathbf{\Gamma}$. The transition graph associated with this Markov chain is depicted in Fig.~\ref{fig:markov_diagram}.

\begin{figure}[!h]
\begin{center}
\includegraphics[width=6.5cm,keepaspectratio]{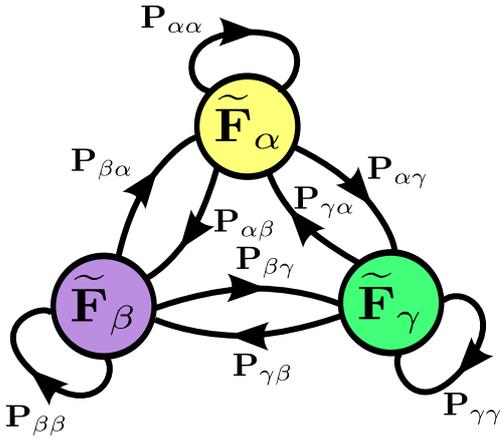}
\caption{(Color online) Transition graph of the Markov chain that follows the dipole force in the dressed states basis.}
\label{fig:markov_diagram}
\end{center}
\end{figure} 

Since there is a unique stationary population distribution $\widetilde{\boldsymbol\rho}^{\infty}$, the Prop.~\ref{prop:D_X} in App.~\ref{ann:diff_coef_markov_proc} gives the diffusion coefficient for the momentum dispersion associated with the stochastic evolution of the dipole force,
\begin{equation}
D_{\rm dip} = \frac{1}{m^{2}} \left\langle \widetilde{\mathbf{F}}, \boldsymbol\Gamma^{\sharp} \widetilde{\mathbf{F}} \right\rangle,
\end{equation}
where $\widetilde{\mathbf{F}} = ( \widetilde{\mathbf{F}}_{\alpha},\widetilde{\mathbf{F}}_{\beta},\widetilde{\mathbf{F}}_{\gamma} )^{T}$, $\mathbf{\Gamma}^{\sharp}$ is the group inverse of $\mathbf{\Gamma}$, and the scalar product is defined as $\left\langle \mathbf{x}, \mathbf{y} \right\rangle = \sum_{i} \widetilde{\rho}^{\infty}_{ii} x_{i} y_{i}$.

\subsection{\label{ann:diff_coef_markov_proc}Diffusion coefficient of a Markov process}

Generalizing the definition of the diffusion coefficient for the dipole force [Eq.~(\ref{eq:diff_coeff_def})], we define here the diffusion coefficient associated with a homogeneous continuous time Markov chain (HCTMC) and show that it can be efficiently calculated using the group inverse of the infinitesimal generator of the process.

A HCTMC $\left\{ X_{t}, t \geq 0 \right\}$ is characterized by its transition matrix $\mathbf{P} \left( \tau \right)$ with elements $\mathbf{P}_{ij} \left( \tau \right) = P \left( X_{\tau} = x_{j} | X_{0} = x_{i} \right)$. It is a stochastic matrix: $\sum_{i} \mathbf{P}_{ij} \left( \tau \right) = 1$, which obeys the Kolmogorov forward equation
\begin{equation}
\frac{d}{d\tau} \mathbf{P} \left( \tau \right) + \mathbf{G} \mathbf{P} \left( \tau \right) = \mathbf{0},
\end{equation}
where $\mathbf{G}$ is called the infinitesimal generator of the HCTMC. Since the initial condition $\mathbf{P} (0) = \mathds{1}$ must be verified, the solution of this differential equation is $\mathbf{P} \left( \tau \right) = \exp \left( - \mathbf{G} \tau \right)$. The generator satisfies $\sum_{i} \mathbf{G}_{ij}=0$ and all the off-diagonal elements must be negative.

The analysis of the long-term behavior of the Markov process makes use of the stationary probability vector $\boldsymbol\Pi^{\infty}$ which is the probability to be in a given state of the chain after an infinite time. It satisfies $\boldsymbol\Pi^{\infty} \mathbf{G} = \mathbf{0}$ with $\| \boldsymbol\Pi^{\infty} \|_{1} = 1$. To quantify the fluctuations of the stochastic variable around the stationary distribution we define the diffusion coefficient:

\begin{definition}
Let $X = \left\{ X_{t}, t \geq 0 \right\}$ be an irreducible HCTMC over the countable state space $\mathcal{S} = \left\{ x_{i} \right\}$ with transition matrix $\mathbf{P} \left( \tau \right)$. If a stationary distribution $\boldsymbol\Pi^{\infty}$ exists, then the \textit{diffusion coefficient} of $X$ is defined as
\begin{equation}
D_{X} = \int_{0}^{\infty} \left( \left\langle X_{\infty} X_{\tau} \right\rangle - \left\langle X_{\infty} \right\rangle^{2} \right) \; d\tau,
\label{eq:diff_coeff_def2}
\end{equation}
where
\begin{eqnarray}
\left\langle X_{\infty} X_{\tau} \right\rangle & = & \sum_{i,j\in \mathcal{S}} x_{i} x_{j} \Pi^{\infty}_{i} \mathbf{P}_{ij} \left( \tau \right), \\
\left\langle X_{\infty} \right\rangle & = & \sum_{i\in \mathcal{S}} x_{i} \Pi^{\infty}_{i}.
\end{eqnarray}
\end{definition}


The transition matrix elements must satisfy $\lim_{\tau \rightarrow \infty} \mathbf{P}_{ij} (\tau) < \infty$, which means that any eigenvalue of $\mathbf{G}$ must be positive. Moreover, if all the eigenvalues are strictly positive then $\lim_{\tau \rightarrow \infty} \mathbf{P} (\tau) = \mathbf{0}$ which is not a stochastic matrix. As a consequence the generator must have at least one eigenvalue equal to zero, meaning that $\mathbf{G}$ is \textit{singular}.

These properties -- which are a consequence of the Perron-Frobenius theorem -- allow us to relate the diffusion coefficient to the generator of the Markov chain. More specifically, the diffusion coefficient is the \textit{quadratic form} involving the states of the chain and associated with the group inverse (a special kind of pseudoinverse) of $\mathbf{G}$. This result is particularly useful for the numerical calculation of the diffusion coefficient.

\begin{proposition}
Let $X=\left\{ X_{t}, t \geq 0 \right\}$ be an irreducible HCTMC over the countable state space $\mathcal{S} = \left\{ x_{i} \right\}$ which has a stationary distribution $\boldsymbol\Pi^{\infty}$. If $\mathbf{G}$ is the infinitesimal generator of $X$, then the diffusion coefficient is
\begin{equation}
D_{X} = \left\langle \mathbf{x}, \mathbf{G}^{\sharp} \mathbf{x} \right\rangle,
\label{eq:diff_coeff_prop}
\end{equation}
where $\mathbf{x} = \left( x_{i} \right)$ is the column vector of the Markov chain states, $\mathbf{G}^{\sharp}$ is the group inverse of $\mathbf{G}$, and the scalar product is defined as $\left\langle \mathbf{x}, \mathbf{y} \right\rangle = \sum_{i\in \mathcal{S}} \Pi_{i}^{\infty} x_{i} y_{i}$.
\label{prop:D_X}
\end{proposition}
\begin{proof}
Since $\left\langle X_{\infty} \right\rangle^{2} = \sum_{i,j\in \mathcal{S}} x_{i} x_{j} \Pi_{i}^{\infty} \Pi_{j}^{\infty}$, it follows that
\begin{equation}
\left\langle X_{\infty} X_{\tau} \right\rangle - \left\langle X_{\infty} \right\rangle^{2} = \sum_{i,j\in \mathcal{S}} x_{i} x_{j} \Pi_{i}^{\infty} \left( \mathbf{P}_{ij} (\tau) - \Pi_{j}^{\infty} \right).
\end{equation}
Moreover, $X$ is irreducible and has a stationary distribution, and therefore the theorem of convergence to invariant distribution for a continuous time Markov chain is satisfied, and $\mathbf{P}_{ij} (\infty) \equiv \lim_{\tau \rightarrow \infty} \mathbf{P}_{ij} (\tau) = \Pi_{j}^{\infty}$, which results in
\begin{equation}
D_{X} = \sum_{i,j\in \mathcal{S}} x_{i} x_{j} \Pi_{i}^{\infty} \mathbf{Z}_{ij}=\langle \mathbf{x}, \mathbf{Z} \mathbf{x} \rangle,
\end{equation}
where $\mathbf{Z} = \int_{0}^{\infty} \left[ \mathbf{P} (\tau) - \mathbf{P} (\infty) \right] d\tau$. Finally, Lemma~\ref{lem:lemma} provides the relation between the fundamental matrix and the generator of the chain: $\mathbf{Z}=\mathbf{G}^{\sharp}$. As a result the diffusion coefficient is $D_{X} = \langle \mathbf{x}, \mathbf{G}^{\sharp} \mathbf{x} \rangle$.
\end{proof}


The calculation of the diffusion coefficient results from the calculation of the so-called \textit{fundamental matrix} or \textit{deviation matrix} $\mathbf{Z}$ \cite{so62339,Ham2004}, which is related to the expected time spent in state $j$ starting from $i$. It is equal to the group inverse of the Markov chain generator.

\begin{lemma}
\begin{equation}
\mathbf{Z} \equiv \int_{0}^{\infty} \left[ \mathbf{P} (\tau) - \mathbf{P} (\infty) \right] \; d\tau = \mathbf{G}^{\sharp}.
\end{equation}
\label{lem:lemma}
\end{lemma}
\begin{proof}
By the Perron-Frobenius theorem, the generator admits the following eigendecomposition:
\begin{equation}
\mathbf{G} = \mathbf{V} \left[ \mathbf{0} \oplus \boldsymbol\Delta \right] \mathbf{V}^{-1},
\end{equation}
with $\boldsymbol\Delta = \mathrm{diag} \left( \lambda_{1}, \dots, \lambda_{n} \right)$, $n$ being the number of nonzero eigenvalues and $\lambda_{i}>0$. It results in
\begin{equation}
\mathbf{P} (\tau) = e^{-\mathbf{G}\tau} = \mathbf{V} \left[ \mathds{1} \oplus e^{-\boldsymbol\Delta \tau} \right] \mathbf{V}^{-1},
\end{equation}
and since $\lambda_{i}>0$,
\begin{equation}
\mathbf{P} (\infty) = \lim_{\tau \rightarrow \infty} \mathbf{V} \left[ \mathds{1} \oplus e^{-\boldsymbol\Delta \tau} \right] \mathbf{V}^{-1} = \mathbf{V} \left[ \mathds{1} \oplus \mathbf{0} \right] \mathbf{V}^{-1}.
\end{equation}
The deviation from the stationary probability is thus
\begin{equation}
\mathbf{P} (\tau) - \mathbf{P} (\infty) = \mathbf{V} \left[ \mathbf{0} \oplus e^{-\boldsymbol\Delta \tau} \right] \mathbf{V}^{-1}.
\end{equation}
We can now evaluate the integral
\begin{eqnarray}
\mathbf{Z}_{ij} & = & \int_{0}^{\infty} \left( \mathbf{V} \left[ \mathbf{0} \oplus e^{-\boldsymbol\Delta \tau} \right] \mathbf{V}^{-1} \right)_{ij} \; d\tau \\
& = & \int_{0}^{\infty} \sum_{ll'} \mathbf{V}_{il} \left[ \mathbf{0} \oplus e^{-\boldsymbol\Delta \tau} \right]_{ll'} \mathbf{V}^{-1}_{l'j} \; d\tau \\
& = & \sum_{ll'} \mathbf{V}_{il} \left[ \mathbf{0} \oplus \int_{0}^{\infty} e^{-\boldsymbol\Delta \tau} \; d\tau \right]_{ll'} \mathbf{V}^{-1}_{l'j}.
\end{eqnarray}
Using again the fact that $\lambda_{i}>0$, it follows that
\begin{equation}
\int_{0}^{\infty} e^{-\boldsymbol\Delta \tau} \; d\tau = \boldsymbol\Delta^{-1} = \mathrm{diag} \left(\lambda_{1}^{-1},\dots,\lambda_{n}^{-1} \right),
\end{equation}
and finally
\begin{equation}
\mathbf{Z}_{ij} = \sum_{ll'} \mathbf{V}_{il} \left[ \mathbf{0} \oplus \boldsymbol\Delta^{-1} \right]_{ll'} \mathbf{V}^{-1}_{l'j} =  \left( \mathbf{V} \left[ \mathbf{0} \oplus \boldsymbol\Delta^{-1} \right] \mathbf{V}^{-1} \right)_{ij}.
\end{equation}
Writing $\boldsymbol\Sigma = \mathbf{0} \oplus \boldsymbol\Delta$ and $\boldsymbol\Sigma^{\sharp} = \mathbf{0} \oplus \boldsymbol\Delta^{-1}$ we have $\mathbf{G} = \mathbf{V} \boldsymbol\Sigma \mathbf{V}^{-1}$ and $\mathbf{Z} = \mathbf{V} \boldsymbol\Sigma^{\sharp} \mathbf{V}^{-1}$. Clearly $\boldsymbol\Sigma \boldsymbol\Sigma^{\sharp} \boldsymbol\Sigma = \boldsymbol\Sigma$, $\boldsymbol\Sigma^{\sharp} \boldsymbol\Sigma \boldsymbol\Sigma^{\sharp} = \boldsymbol\Sigma^{\sharp}$ and $\boldsymbol\Sigma^{\sharp} \boldsymbol\Sigma = \boldsymbol\Sigma \boldsymbol\Sigma^{\sharp}$, therefore
\begin{eqnarray}
\mathbf{G} \mathbf{Z} \mathbf{G} = \mathbf{V} \boldsymbol\Sigma \boldsymbol\Sigma^{\sharp} \boldsymbol\Sigma \mathbf{V}^{-1} = \mathbf{V} \boldsymbol\Sigma  \mathbf{V}^{-1} = \mathbf{G}, \\
\mathbf{Z} \mathbf{G} \mathbf{Z} = \mathbf{V} \boldsymbol\Sigma^{\sharp} \boldsymbol\Sigma \boldsymbol\Sigma^{\sharp} \mathbf{V}^{-1} = \mathbf{V} \boldsymbol\Sigma^{\sharp}  \mathbf{V}^{-1} = \mathbf{Z}, \\
\mathbf{G} \mathbf{Z} = \mathbf{V} \boldsymbol\Sigma \boldsymbol\Sigma^{\sharp}  \mathbf{V}^{-1} = \mathbf{V} \boldsymbol\Sigma^{\sharp} \boldsymbol\Sigma  \mathbf{V}^{-1} = \mathbf{Z} \mathbf{G}.
\end{eqnarray}
In other words, $\mathbf{Z}$ has all the properties of the group inverse of $\mathbf{G}$, and a matrix that satisfies this properties is unique: $\mathbf{Z} = \mathbf{G}^{\sharp}$.
\end{proof}


As an example, we can use this result to calculate the diffusion coefficient associated with the dipole force fluctuation in a two-level atom. 
\begin{example}
In the case of the dipolar force on a two-level system, one as $\mathcal{S} = \left\{ +F, -F \right\}$ where $F = \hbar \nabla \Omega /2$ and the transitions between the dressed states are described by the following generator \cite{dalibard1985}:
\begin{equation}
\mathbf{G} = \Gamma \left( 
\begin{array}{cc}
\cos^{4} \theta  & -\cos^{4} \theta \\
-\sin^{4} \theta & \sin^{4} \theta
\end{array}
\right),
\end{equation}
where the angle $\theta$ satisfies $\tan 2 \theta = \Omega / \Delta$. The stationary population vector is 
\begin{equation}
\boldsymbol\Pi^{\infty} = \frac{\Gamma}{\Gamma_{\rm pop}} \left(
\begin{array}{c}
\sin^{4} \theta \\
\cos^{4} \theta
\end{array}
\right), 
\end{equation}
where $\Gamma_{\rm pop} = \Gamma \left( \cos^{4} \theta + \sin^{4} \theta \right)$. The generator has the following eigendecomposition: $\mathbf{G} = \mathbf{V} \boldsymbol\Delta \mathbf{V}^{-1}$ with $\boldsymbol\Delta = \mathrm{diag} \left( 0, \Gamma_{\rm pop} \right)$ and
\begin{equation}
\mathbf{V} = \left( 
\begin{array}{cc}
1 & -\cot^{4} \theta \\
1 & 1
\end{array}
\right), \; \mathbf{V}^{-1} = \frac{\Gamma}{\Gamma_{\rm pop}} \left( 
\begin{array}{cc}
\sin^{4} \theta & \cos^{4} \theta \\
-\sin^{4} \theta & \sin^{4} \theta
\end{array}
\right).
\end{equation}
The group inverse of the generator is thus
\begin{equation}
\mathbf{G}^{\sharp} = \frac{\Gamma}{\Gamma_{\rm pop}^{2}} \left( 
\begin{array}{cc}
\cos^{4} \theta & -\cos^{4} \theta \\
-\sin^{4} \theta & \sin^{4} \theta
\end{array}
\right).
\end{equation}
The diffusion coefficient is finally obtained from Eq.~(\ref{eq:diff_coeff_prop}) with $\mathbf{x} = (+F,-F)^{T}$:
\begin{equation}
D_{\rm dip} = \left\langle \mathbf{x}, \mathbf{G}^{\sharp} \mathbf{x} \right\rangle = \frac{4 F^{2}}{\Gamma} \frac{\sin^{4} \theta \cos^{4} \theta}{\left( \sin^{4} \theta + \cos^{4} \theta \right)^{3}}.
\label{eq:diff_coeff_tls}
\end{equation}
In the high-intensity limit ($\Omega \gg \Delta$), the angle is $\theta = \pi/4$ and the diffusion coefficient becomes $D_{\rm dip} = \hbar^{2} |\nabla \Omega|^{2} / 2 \Gamma$. Conversely, in the low-intensity regime ($\Omega \ll \Delta$) the angle is small ($\theta \sim 0$) and the diffusion coefficient vanishes.
\end{example}

\end{document}